\documentclass[11pt,letterpaper]{article}

\usepackage[T1]{fontenc}
\usepackage{lmodern}
\usepackage{microtype}
\usepackage[margin=1in]{geometry}
\usepackage{amsmath,amssymb,amsthm,mathtools,bm}
\usepackage{booktabs}
\usepackage{array}
\usepackage{float}
\usepackage{graphicx}
\usepackage{enumitem}
\usepackage{xcolor}
\usepackage{xurl}
\usepackage{hyperref}
\usepackage[nameinlink,noabbrev]{cleveref}

\hypersetup{
  colorlinks=true,
  linkcolor=blue!65!black,
  citecolor=blue!65!black,
  urlcolor=blue!65!black,
  pdftitle={Efficient Exact Quantum Sampling from the Sun--Wootters Distribution for Optimal Polynomial Intersection},
  pdfauthor={Sunghyeon Jo}
}

\newtheorem{theorem}{Theorem}[section]
\newtheorem{lemma}[theorem]{Lemma}
\newtheorem{proposition}[theorem]{Proposition}
\newtheorem{corollary}[theorem]{Corollary}
\theoremstyle{definition}
\newtheorem{definition}[theorem]{Definition}
\newtheorem{remark}[theorem]{Remark}

\newcommand{\F}{\mathbb{F}}
\newcommand{\C}{\mathcal{C}}
\newcommand{\Kcal}{\mathcal{K}}
\newcommand{\Lcal}{\mathcal{L}}

\newcommand{\supp}{\operatorname{supp}}
\newcommand{\wt}{\operatorname{wt}}
\newcommand{\poly}{\operatorname{poly}}
\newcommand{\SCL}{\operatorname{SCL}}
\newcommand{\TV}{\operatorname{TV}}
\newcommand{\E}{\mathbb{E}}
\newcommand{\Prb}{\mathbb{P}}
\newcommand{\ket}[1]{\lvert #1\rangle}

\newcommand{\abs}[1]{\left\lvert #1\right\rvert}
\newcommand{\norm}[1]{\left\lVert #1\right\rVert}
\newcommand{\one}{\mathbf{1}}

\newcommand{\QFT}{\operatorname{QFT}}
\newcommand{\Had}{\operatorname{Had}}

\title{\textbf{Efficient Exact Quantum Sampling from the Sun--Wootters Distribution for Optimal Polynomial Intersection}}
\author{Sunghyeon Jo}
\date{\today}

\begin{document}
\hypersetup{pageanchor=false}
\maketitle

\begin{abstract}
Optimal Polynomial Intersection (OPI) is a structured optimization problem for which Decoded Quantum Interferometry (DQI) attains a satisfaction guarantee governed by the semicircle law. Sun and Wootters recently showed that, for balanced OPI over prime fields, a Fourier-defined distribution $P_u$ gives a strict worst-case improvement from limiting rate $0.6225$ onward and asymptotically perfect solutions from rate $0.7496$ onward. They asked whether $P_u$ can be sampled efficiently.

We answer this question for Reed--Solomon OPI parameters satisfying the Sun--Wootters exponent condition strictly below the dual Johnson radius. Under coherent membership-oracle access, we give a bounded-error polynomial-time quantum sampler for $P_u$. The ideal circuit samples $P_u$ exactly conditioned on success, while a finite-precision implementation achieves any prescribed inverse-polynomial total-variation error. Consequently, every fixed limiting rate $0.6225\le r<1$ admits a strict worst-case improvement over the DQI semicircle value, and every limiting rate $r\ge 3/4$ admits solutions of satisfaction $1-o(1)$ with high probability.

The algorithm coherently sums the amplitudes of all low-weight errors in each syndrome class using deterministic complete list decoding. Complete Reed--Solomon list decoding and the Sun--Wootters denominator estimate make both the list size and the postselection overhead polynomial. In concurrent and independent work, Horinaga and Yamakawa obtain worst-case OPI algorithms over prime-power fields and exact satisfaction at every fixed rate strictly above $3/4$.
\end{abstract}

\thispagestyle{empty}
\clearpage
\tableofcontents
\thispagestyle{empty}
\clearpage
\hypersetup{pageanchor=true}
\pagenumbering{arabic}

\section{Introduction}

Decoded Quantum Interferometry (DQI) combines coherent decoding with a quantum Fourier transform to solve structured optimization problems \cite{JordanEtAl2025}.  Its canonical algebraic example is Optimal Polynomial Intersection (OPI).  An OPI instance specifies distinct points $\alpha_1,\ldots,\alpha_m\in\F_p$ and subsets $S_1,\ldots,S_m\subseteq\F_p$.  The task is to find a polynomial $Q$ of degree less than $n$ for which $Q(\alpha_i)\in S_i$ holds for as many indices as possible.

The original OPI algorithm decodes an error relative to the dual Reed--Solomon code.  It is correct within the unique-decoding radius, and in the balanced case its expected satisfaction ratio follows a semicircle law.  Subsequent work gave nearly linear implementations in random-access models \cite{KhattarEtAl2025,Rosmanis2026}.  The attainable score is therefore governed by the decoding radius.

Sun and Wootters strictly improved the worst-case existential benchmark beyond the semicircle law \cite{SunWootters2026}.  Over prime fields and at balanced density, they prove an improvement from rate $n/m=0.6225$ onward and asymptotically perfect solutions from rate $0.7496$ onward.  Their Question~1.17 asks whether this improvement can be realized by a polynomial-time quantum algorithm.  The discussion following that question proposes a Fourier-defined distribution $P_u$ whose expected score obeys their improved bound, leaving efficient sampling as the algorithmic problem.

We give such a sampler for every target cutoff strictly below the dual Reed--Solomon Johnson radius.  The central operation is a coherent sum over all low-weight errors with the same syndrome.  Unlike choosing one element of a decoded list, this operation preserves the complex phase and cancellation pattern that defines $P_u$.

\subsection{Our Results}

Our first result is a finite, exact circuit identity.  Suppose that a deterministic complete list decoder enumerates every candidate in each radius-$\ell$ syndrome class, also called a syndrome fiber.  Conditioned on a successful list-index measurement, the ideal circuit prepares the normalized Sun--Wootters syndrome state exactly.  The success probability is
\[
   p_{\mathrm{fib}}=\frac{Z_I}{LZ_0},
\]
where $Z_0$ is the squared norm before fiber summation, $Z_I$ is the squared norm after fiber summation, and $L$ is one global padded list length.  Using a single value of $L$ for all syndromes is essential: fiber-dependent normalization would change the target distribution.

The dual evaluation code in OPI is a generalized Reed--Solomon code.  The deterministic list-decoding algorithm of Chatterjee, Harsha, and Kumar \cite{ChatterjeeEtAl2026} therefore makes the construction polynomial time whenever
\begin{equation}
   m-\ell>\sqrt{(m-n-1)m}.
   \label{eq:intro-johnson}
\end{equation}
The Sun--Wootters denominator estimate then gives $Z_I/Z_0=1+o(1)$ under their exponent condition.  Consequently, the flagged outcome occurs with probability $\Theta(1/L)$.

The resulting parameter consequences are summarized below.  The formal sampler theorem is \cref{thm:main-sampler}; the two balanced-case consequences are \cref{cor:rate-06225,cor:perfect-three-quarter}.

\begin{theorem}[Informal; balanced prime-field case]
\label{thm:intro-main}
Consider a sequence of prime-field Reed--Solomon OPI instances with $m\to\infty$, $m\le p$, distinct evaluation points, common density $\rho_m=1/2+o(1)$, and coherent membership-oracle access.
\begin{enumerate}[label=\textup{(\roman*)},leftmargin=2.2em]
  \item If $n/m\to r$ for a fixed $0.6225\le r<1$, there is a uniform polynomial-time quantum algorithm whose expected satisfaction ratio exceeds the DQI semicircle value by a positive constant depending on $r$.
  \item If $n/m\to r\ge3/4$, there is a uniform polynomial-time quantum algorithm that outputs a polynomial of satisfaction ratio $1-o(1)$ with probability $1-o(1)$.
\end{enumerate}
Whenever the target distribution is well defined and a complete list is available, the ideal successful branch samples the Sun--Wootters distribution defined by that radius exactly.  A finite-precision implementation samples it to any prescribed inverse-polynomial total-variation error.
\end{theorem}

The first item gives an efficient sampler for the specific distribution proposed by Sun and Wootters and begins at balanced rate $0.6225$; the concurrent algorithm discussed below guarantees a strict improvement for every fixed rate above $0.6979$.  The second item includes the limiting Johnson point $r=3/4$ and gives asymptotically perfect satisfaction there.

\Cref{fig:balanced-bounds} compares the balanced-case algorithmic bounds.  Below rate $3/4$, the curve for this paper is the supremum of \cref{eq:scl} over parameters satisfying the Sun--Wootters condition \cref{eq:SW-condition} and the Johnson-radius bound.  The Horinaga--Yamakawa curve is the boundary of the two conditions in their informal balanced-case theorem \cite{HorinagaYamakawa2026}.  Fixed points below either boundary satisfy the required strict inequalities.

\begin{figure}[H]
\centering
\includegraphics[width=0.94\linewidth]{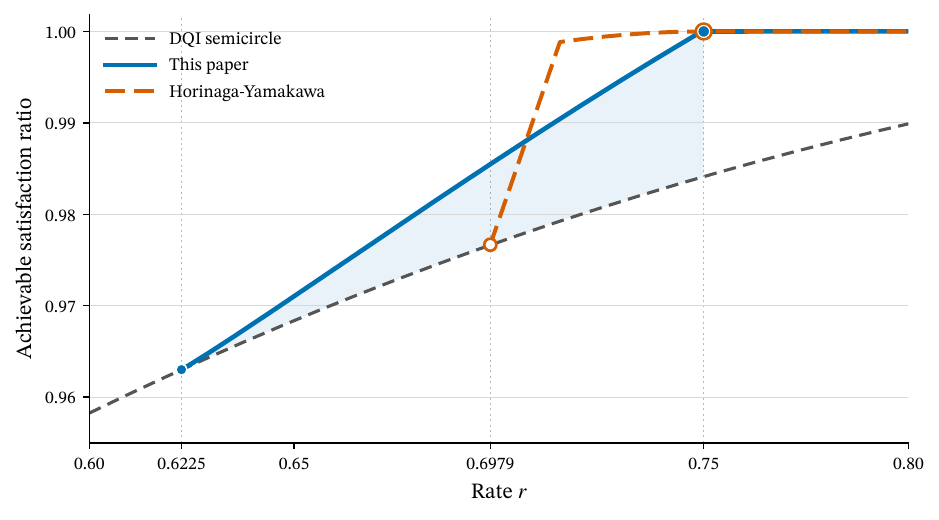}
\caption{Asymptotic algorithmic satisfaction bounds in the balanced case.  The curves show supremal boundary values.  Blue markers identify the included rates $0.6225$ and $3/4$; orange rings identify the reported sufficient rates $0.6979$ and $3/4$ for Horinaga--Yamakawa, whose guarantees apply strictly to their right.  At rates at least $3/4$, this paper gives satisfaction $1-o(1)$ with high probability, while Horinaga--Yamakawa give satisfaction $1$ at every fixed rate strictly above $3/4$.}
\label{fig:balanced-bounds}
\end{figure}

\subsection{Technical Overview}

\paragraph{The Fourier-side state.}
The DQI state preparation produces a superposition of low-weight errors and their syndromes of the form
\begin{equation}
   \frac1{\sqrt{Z_0}}
   \sum_y a_y\ket{y}\ket{Hy}.
   \label{eq:intro-labeled-state}
\end{equation}
Here $H=B^{\mathsf T}$ is a parity-check map for the dual evaluation code, and the complex coefficient $a_y$ is a product of local Fourier coefficients.  Within the unique-decoding radius, the syndrome determines the error and the error register can be erased coherently.  At larger radii, several errors may have the same syndrome.  The Sun--Wootters state requires the coherent sum of their weighted amplitudes; list cardinality alone contains no phase information.

\paragraph{Coherent fiber summation.}
We run a complete deterministic list decoder, sort its output canonically, and pad every list to the same power-of-two length $L$.  A reversible indexer maps each supported pair $(y,s)$ to the canonical index of $y$ in the list for syndrome $s$.  A match flag distinguishes a genuine element in slot zero from an invalid input.  Projecting the index register onto the uniform state then implements
\begin{equation}
   \sum_{y:\,Hy=s}a_y\ket{y}
   \longmapsto
   \left(\sum_{y:\,Hy=s}a_y\right)\ket{s}
   \label{eq:intro-fiber-sum}
\end{equation}
with the same multiplicative factor $L^{-1/2}$ for every syndrome.  Thus all phases and cancellations are retained exactly.  After the inverse Fourier transform, the resulting measurement distribution is precisely $P_u$ in the ideal circuit.

\paragraph{Polynomial complexity.}
Inside the Johnson radius, an elementary pair-counting argument bounds every decoded list by $m^2$; with fixed normalized slack, the list size is constant.  The deterministic classical decoder can be padded to a fixed running time, compiled reversibly, and uncomputed.  The denominator estimate of Sun and Wootters bounds the change in squared norm from interference by $e^{-\Omega(m)}$.  Repeating the circuit $O(L\log(1/\beta))$ times therefore gives success probability at least $1-\beta$.  The finite-precision analysis separately treats state preparation, conditioning, and the final prime-field Fourier transform.

\paragraph{From sampling to optimization.}
Sun and Wootters analyze a radius $\mu+\delta$ beyond the DQI cutoff $\mu=n/(2m)$.  Whenever their exponent condition is strict, decreasing $\delta$ preserves it.  A strict decrease also places the radius inside Johnson.  This monotonicity transfers every balanced strict-improvement point to an efficiently sampleable point with a still-positive score gap.  At limiting rate $3/4$, a finite-length cutoff approaches radius $1/2$ from inside Johnson and yields the second item of \cref{thm:intro-main}.

\subsection{Related and Concurrent Work}
\label{sec:related-work}

\paragraph{Fourier decoding algorithms.}
The use of coherent decoding in the Fourier domain originates in Regev's quantum reduction and was developed further in the quantum-decoding and verifiable-advantage literature \cite{Regev2009,YamakawaZhandry2024,ChaillouxTillich2024}.  Chailloux and Tillich introduced soft-decoder methods that tolerate inverse-polynomial decoding success, and subsequent work applied this perspective to improved average-case OPI algorithms \cite{ChaillouxTillich2025,Chailloux2025}.  DQI gives the worst-case OPI algorithm at the semicircle bound \cite{JordanEtAl2025}; Khattar et al. and Rosmanis study its efficient implementation \cite{KhattarEtAl2025,Rosmanis2026}.

\paragraph{Sun--Wootters.}
Sun and Wootters provide the existential, denominator, and performance estimates used here \cite{SunWootters2026}.  The present circuit realizes their phase-sensitive syndrome sums in polynomial time for cutoffs below the dual Reed--Solomon Johnson radius.

\paragraph{Concurrent work.}
Concurrent and independent work of Horinaga and Yamakawa gives a different worst-case quantum algorithm beyond DQI \cite{HorinagaYamakawa2026}.  Their Regev-like construction uses an approximately uniform member of a decoded list.  In the balanced case $\rho=1/2$, they show $R_0^{\mathrm{alg}}(1/2)<0.6979$, and hence guarantee an improvement for every fixed rate above $0.6979$; they obtain satisfaction $1$ at every fixed rate above $3/4$.  Under the assumptions $q=p^e$, $p=\omega(1)$, $\log q\le\poly(m)$, and the constant-margin hypotheses of their theorem, their results also cover general density $\rho\in(0,1)$ and MDS MaxLINSAT instances whose dual code admits an efficient list decoder.  The present algorithm samples the Sun--Wootters distribution over prime fields, reaches the balanced threshold $0.6225$, and includes the limiting point $3/4$.

\paragraph{Parameter range.}
The theorem gives polynomial complexity using a reversible compilation of the deterministic list decoder.  It applies to prime-field Reed--Solomon OPI with $\rho=1/2+o(1)$ and coherent membership-oracle access.  Subject to their asymptotic field and constant-margin assumptions, Horinaga and Yamakawa cover prime-power fields, general $\rho\in(0,1)$, and MDS MaxLINSAT when the dual code admits an efficient list decoder.  Exact Sun--Wootters sampling beyond the Johnson radius calls for a coherent weighted-sum algorithm whose complexity depends on amplitudes more directly than on maximum list size.

\section{Preliminaries}

\subsection{Fourier Analysis over a Prime Field}

Let $p$ be prime and let $\omega=e^{2\pi i/p}$.  For $f:\F_p\to\mathbb C$, define
\begin{equation}
   \widehat f(z):=\frac1p\sum_{a\in\F_p} f(a)\omega^{-az},
   \qquad
   f(a)=\sum_{z\in\F_p}\widehat f(z)\omega^{az}.
   \label{eq:fourier-convention}
\end{equation}
Parseval's identity is
\begin{equation}
   \sum_{z\in\F_p}\abs{\widehat f(z)}^2
   =\frac1p\sum_{a\in\F_p}\abs{f(a)}^2.
   \label{eq:parseval-one-dimensional}
\end{equation}
For $s,x\in\F_p^n$, write $\chi_s(x)=\omega^{\langle s,x\rangle}$.  Expectations over a finite set are uniform unless a distribution is displayed explicitly.

\subsection{Balanced OPI}

Let $\alpha_1,\ldots,\alpha_m\in\F_p$ be distinct, with $m\le p$, and let
\begin{equation}
   B=\begin{pmatrix}
      1&\alpha_1&\cdots&\alpha_1^{n-1}\\
      \vdots&\vdots&&\vdots\\
      1&\alpha_m&\cdots&\alpha_m^{n-1}
   \end{pmatrix}\in\F_p^{m\times n}.
   \label{eq:vandermonde}
\end{equation}
For $x\in\F_p^n$, the vector $Bx$ consists of evaluations of the polynomial with coefficient vector $x$.  Suppose that all input sets have a common density
\[
   \abs{S_i}=\rho p,
   \qquad 0<\rho<1.
\]
Define the satisfaction ratio
\begin{equation}
   s(x):=\frac1m\sum_{i=1}^m\one\{\langle b_i,x\rangle\in S_i\},
   \label{eq:score}
\end{equation}
where $b_i$ is row $i$ of $B$, and define the normalized discrepancy function
\begin{equation}
   g_i(a):=\frac{\one\{a\in S_i\}-\rho}{\sqrt{\rho(1-\rho)}}.
   \label{eq:gi}
\end{equation}
Then
\begin{equation}
   \frac1p\sum_a g_i(a)=0,
   \qquad
   \frac1p\sum_a g_i(a)^2=1,
   \qquad
   \widehat g_i(0)=0,
   \qquad
   \sum_{z\ne0}\abs{\widehat g_i(z)}^2=1.
   \label{eq:gi-properties}
\end{equation}
For $0\le k\le m$, set
\begin{equation}
   q_k(x):=\sum_{\substack{T\subseteq[m]\\\abs{T}=k}}
              \prod_{i\in T}g_i(\langle b_i,x\rangle).
   \label{eq:qk}
\end{equation}

We study asymptotic families with $m\to\infty$, $1\le n<m\le p$, $n/m\to r=2\mu\in(0,1]$, prime $p=p(m)$, and $\rho=\rho_m=1/2+o(1)$.  All asymptotic constants may depend on the fixed limiting parameters but not on the particular input sets.  For probability distributions $P,Q$ on the same finite set, we use
\[
   \TV(P,Q):=\frac12\sum_x\abs{P(x)-Q(x)}.
\]

\paragraph{Computational model.}
The integers $p,m,n$, the distinct evaluation points, and the common cardinality $\rho p$ are classical input data.  The sets are accessed by the coherent membership oracle in \cref{eq:membership-oracle}.  A deterministic classical preprocessor may use the public input data and the requested accuracy to compute field constants, the right inverse $R$, the multipliers $\nu_i$, and a description of the quantum circuit; its polynomial running time is included in the total complexity.  The resulting oracle circuit is uniform in the standard sense that this description is produced in time polynomial in the classical input length.

Put $b=\lceil\log_2p\rceil$ and identify $\F_p$ with the code subspace $\mathcal H_p=\operatorname{span}\{\ket0,\ldots,\ket{p-1}\}\subseteq(\mathbb C^2)^{\otimes b}$.  Every reversible field operation is specified as a permutation on all $2^b$ binary strings; encodings $p,\ldots,2^b-1$ have fixed total behavior and are never populated by the ideal preparation.  Likewise, an approximate prime-modulus QFT is a qubit unitary whose restriction to $\mathcal H_p$ is within the stated operator-norm error of the mathematical $p$-dimensional QFT, with an arbitrary unitary extension on $\mathcal H_p^\perp$.  Running time counts one- and two-qubit gates, reversible bit operations, and oracle calls.  Arbitrary rotations and the prime-modulus QFT are synthesized over a fixed universal gate set to the precision specified in \cref{cor:finite-precision}.

\subsection{The Sun--Wootters Target Distribution}

Fix $\delta\in[0,1/2-\mu]$ and define
\begin{equation}
   \ell=\lfloor(\mu+\delta)m\rfloor,
   \qquad
   \sigma=\lfloor\log\log\ell\rfloor,
   \qquad
   \Kcal=\{\ell-\sigma,\ldots,\ell\}\cap\{0,\ldots,m\}.
   \label{eq:shell}
\end{equation}
Choose
\begin{equation}
   u_k:=\binom{m}{k}^{-1/2}\one\{k\in\Kcal\}.
   \label{eq:weights}
\end{equation}
The unnormalized target amplitude and its probability distribution are
\begin{equation}
   F(x):=\sum_{k\in\Kcal}u_kq_k(x),
   \qquad
   P_u(x):=\frac{\abs{F(x)}^2}{\sum_{z\in\F_p^n}\abs{F(z)}^2}.
   \label{eq:Pu}
\end{equation}
Because every $g_i$ is real, $F(x)$ is real, so this agrees with the square used by Sun and Wootters.

For $0\le t\le1/2$, the balanced semicircle function is
\begin{equation}
   \SCL_{1/2}(t)=\frac12+\sqrt{t(1-t)}.
   \label{eq:scl}
\end{equation}

We use the following consequence of the proof of Sun and Wootters's Theorem~1.6, specifically their denominator expansion, Lemma~4.11, and Lemma~5.2 \cite{SunWootters2026}.

\begin{theorem}[Sun--Wootters analytic estimate]
\label{thm:SW-import}
Let $r=2\mu$, let $\rho_m=1/2+o(1)$, and use \cref{eq:shell,eq:weights,eq:Pu}.  Define the binary entropy with natural logarithms by
\[
   h(t)=-t\log t-(1-t)\log(1-t),
\]
and
\begin{align}
   E(\mu,\delta)
   &:=(\mu+\delta)h\!\left(\frac{\mu}{\mu+\delta}\right)
      +(1-\mu-\delta)h\!\left(\frac{\mu}{1-\mu-\delta}\right)
      -h(2\mu),
      \label{eq:E-exponent}\\
   G(\mu,\lambda)
   &:=(1-2\mu)\log2+h(2\mu)
      -(4\mu-1)h\!\left(\frac{\lambda}{4\mu-1}\right)\notag\\
   &\hspace{2.8em}
      -(2-4\mu)h\!\left(\frac{2\mu-\lambda}{2-4\mu}\right)
      +\lambda\log\!\left(\frac2\pi\right),
      \label{eq:G-exponent}
\end{align}
where parameters are restricted to the domain on which the entropy arguments lie in $[0,1]$.
If there is a $\lambda$ such that
\begin{equation}
   E(\mu,\delta)+G(\mu,\lambda)<0,
   \label{eq:SW-condition}
\end{equation}
then there exist constants $c>0$ and $C<\infty$ such that, uniformly over all balanced input sets,
\begin{align}
   \E_{x\in\F_p^n}\abs{F(x)}^2
      &=\abs{\Kcal}+O(m^Ce^{-cm}),
      \label{eq:SW-denominator}\\
   \E_{x\sim P_u}s(x)
      &\ge \SCL_{1/2}(\mu+\delta)-o(1).
      \label{eq:SW-score}
\end{align}
\end{theorem}

\begin{proof}[Derivation from the cited estimates]
Sun and Wootters define the same discrepancy functions and the same distribution in their equations~(4.2), (4.17), and (4.18); the proportionality constant in their $u_k$ is immaterial, and we fix it as in \cref{eq:weights}.  Their Proposition~4.7 evaluates the weight-zero part of the denominator as exactly $\sigma+1=\abs{\Kcal}$.  Under \cref{eq:SW-condition}, their Lemma~5.2, combined with Lemma~4.10 and the proof of Theorem~1.6, bounds every nonzero-weight correction by an exponentially decaying term.  Summing the polynomially many weight and shell-index choices gives the additive $O(m^Ce^{-cm})$ term in \cref{eq:SW-denominator}; this is the denominator displayed in their equation~(4.35).  Their Lemma~4.11 then gives the expectation bound under the distribution $P_u$.  The strict exponent margin absorbs the $o(m)$ terms caused by $\rho_m=1/2+o(1)$ and by integer rounding.  All Fourier estimates in that argument are worst-case bounds over the sets $S_i$, so the constants are uniform over the promised inputs.
\end{proof}

\begin{lemma}[Stability under sublinear parameter perturbations]
\label{lem:SW-stability}
The conclusions of \cref{thm:SW-import} remain valid if $n/(2m)=\mu+o(1)$ and the largest active weight satisfies $\ell/m=\mu+\delta+o(1)$, provided the limiting parameters obey the strict inequality \cref{eq:SW-condition}.  The same is true for any fixed $\lambda$ in the interior of its entropy domain.
\end{lemma}

\begin{proof}
In the derivation above, every binomial coefficient contributes its entropy exponent plus $o(m)$, the balanced Fourier estimate contributes $G(\mu,\lambda)+o(1)$ per coordinate, and the shell width is $o(m)$.  Replacing $\mu$ and $\ell/m-\mu$ by convergent sequences changes the continuous entropy exponents by $o(1)$.  A strict negative limiting margin therefore remains negative by a fixed constant for all sufficiently large $m$, and the same summations yield \cref{eq:SW-denominator,eq:SW-score}.
\end{proof}

\begin{remark}
Theorem~1.6 of Sun and Wootters is stated as an existence theorem.  The stronger distributional conclusion in \cref{eq:SW-score} is what their proof establishes before taking a maximizing point: their Lemma~4.11 directly lower-bounds expectation under $P_u$.  Equation \eqref{eq:SW-denominator} is the denominator estimate appearing in their equation~(4.35), after the common normalization of the $u_k$ is fixed as in \cref{eq:weights}.
\end{remark}

\section{The Fourier-Domain Target State}

Let
\begin{equation}
   H:=B^{\mathsf T}\in\F_p^{n\times m},
   \qquad \C:=\operatorname{im}(B),
   \qquad
   \C^\perp:=\ker H.
   \label{eq:H-code}
\end{equation}
For $y\in\F_p^m$, define
\begin{equation}
   a_y:=u_{\wt(y)}\prod_{i:y_i\ne0}\widehat g_i(y_i),
   \label{eq:ay}
\end{equation}
where $u_k=0$ outside $\Kcal$.  Define
\begin{equation}
   Z_0:=\sum_y\abs{a_y}^2,
   \qquad
   \ket{\Psi_0}:=\frac1{\sqrt{Z_0}}\sum_y a_y\ket{y}\ket{Hy}.
   \label{eq:labeled-state}
\end{equation}
For each syndrome $s\in\F_p^n$, let
\begin{equation}
   \Lcal_s:=\{y\in\F_p^m:Hy=s,\ \wt(y)\in\Kcal\},
   \qquad
   A_s:=\sum_{y\in\Lcal_s}a_y,
   \label{eq:fibers}
\end{equation}
and set
\begin{equation}
   Z_I:=\sum_s\abs{A_s}^2,
   \qquad
   \ket{\widehat I}:=\frac1{\sqrt{Z_I}}\sum_sA_s\ket{s}.
   \label{eq:ideal-syndrome}
\end{equation}

\begin{lemma}[Fourier expansion]
\label{lem:fourier-expansion}
For every $x\in\F_p^n$,
\begin{equation}
   F(x)=\sum_{y\in\F_p^m}a_y\chi_{Hy}(x)
       =\sum_{s\in\F_p^n}A_s\chi_s(x).
   \label{eq:fourier-expansion}
\end{equation}
Consequently,
\begin{equation}
   Z_I=\E_{x\in\F_p^n}\abs{F(x)}^2,
   \label{eq:ZI-parseval}
\end{equation}
and applying the inverse QFT over $\F_p^n$ to \cref{eq:ideal-syndrome} produces a state whose measurement distribution is exactly $P_u$.
\end{lemma}

\begin{proof}
By Fourier inversion and $\widehat g_i(0)=0$,
\begin{align*}
   q_k(x)
   &=\sum_{\abs{T}=k}\prod_{i\in T}
       \left(\sum_{z_i\ne0}\widehat g_i(z_i)
       \omega^{z_i\langle b_i,x\rangle}\right)\\
   &=\sum_{\wt(y)=k}
       \left(\prod_{i:y_i\ne0}\widehat g_i(y_i)\right)
       \omega^{\langle B^{\mathsf T}y,x\rangle}.
\end{align*}
Multiplying by $u_k$ and summing over $k$ proves the first equality in \cref{eq:fourier-expansion}; grouping terms by $s=Hy$ proves the second.  Character orthogonality gives
\[
   \frac1{p^n}\sum_x\abs{F(x)}^2
   =\sum_s\abs{A_s}^2=Z_I.
\]
Finally, under the convention
\[
   (\QFT_p^{-1})^{\otimes n}\ket{s}
      =p^{-n/2}\sum_x\chi_s(x)\ket{x},
\]
the amplitude at $x$ is $F(x)/(p^{n/2}\sqrt{Z_I})$.  After normalization, its squared magnitude gives the probability in \cref{eq:Pu}.
\end{proof}

\begin{lemma}[Norm of the labeled error state]
\label{lem:Z0}
For the weights in \cref{eq:weights},
\begin{equation}
   Z_0=\abs{\Kcal}=\sigma+1
\end{equation}
for all sufficiently large $m$.
\end{lemma}

\begin{proof}
Fix a support $T\subseteq[m]$.  By \cref{eq:gi-properties},
\[
   \sum_{\supp(y)=T}
      \prod_{i\in T}\abs{\widehat g_i(y_i)}^2
   =\prod_{i\in T}\sum_{z\ne0}\abs{\widehat g_i(z)}^2=1.
\]
For weight $k$, there are $\binom mk$ supports and $u_k^2=\binom mk^{-1}$, so every active weight contributes one.
\end{proof}

\section{Complete List Decoding of Syndrome Fibers}

\subsection{The Dual Code Is Generalized Reed--Solomon}

For distinct $\alpha_1,\ldots,\alpha_m$, define
\begin{equation}
   \nu_i:=\left(\prod_{j\ne i}(\alpha_i-\alpha_j)\right)^{-1}.
   \label{eq:dual-multipliers}
\end{equation}

\begin{lemma}[Dual evaluation code]
\label{lem:dual-GRS}
The code $\C^\perp=\ker B^{\mathsf T}$ is
\begin{equation}
   \C^\perp
   =\left\{(\nu_1f(\alpha_1),\ldots,\nu_mf(\alpha_m)):
       f\in\F_p[X],\ \deg f<m-n\right\}.
   \label{eq:dual-GRS}
\end{equation}
In particular it is a generalized Reed--Solomon code of dimension $m-n$ and distance $n+1$.
\end{lemma}

\begin{proof}
Let $q\in\F_p[X]$ have degree at most $m-2$.  Lagrange interpolation writes
\[
   q(X)=\sum_{i=1}^m q(\alpha_i)
        \prod_{j\ne i}\frac{X-\alpha_j}{\alpha_i-\alpha_j}.
\]
The coefficient of $X^{m-1}$ on the left is zero, while on the right it is
$\sum_i\nu_iq(\alpha_i)$.  Therefore
\begin{equation}
   \sum_{i=1}^m\nu_iq(\alpha_i)=0
   \quad\text{whenever }\deg q\le m-2.
   \label{eq:lagrange-orthogonality}
\end{equation}
If $h$ has degree less than $n$ and $f$ has degree less than $m-n$, then $q=hf$ has degree at most $m-2$.  Equation \eqref{eq:lagrange-orthogonality} shows that the vector in \cref{eq:dual-GRS} is orthogonal to every evaluation vector $(h(\alpha_i))_i$, so it lies in $\C^\perp$.  The displayed family has dimension $m-n$, equal to $\dim\C^\perp$, hence equality holds.  The distance statement is the standard root bound: a nonzero polynomial of degree less than $m-n$ has at most $m-n-1$ zeros among the evaluation points.
\end{proof}

\subsection{Fibers as Decoding Lists}

Since $B$ has full column rank, $H=B^{\mathsf T}$ has full row rank.  Choose any right inverse $R\in\F_p^{m\times n}$ satisfying $HR=I_n$, and define the canonical representative
\begin{equation}
   r_s:=Rs.
   \label{eq:representative}
\end{equation}
Then
\begin{equation}
   Hy=s
   \quad\Longleftrightarrow\quad
   y=r_s-c\text{ for a unique }c\in\C^\perp.
   \label{eq:coset-bijection}
\end{equation}
Consequently,
\begin{equation}
   \Lcal_s
   =\{r_s-c:c\in\C^\perp,\ \wt(r_s-c)\in\Kcal\}.
   \label{eq:fiber-list}
\end{equation}
Thus a complete decoder for $\C^\perp$ centered at $r_s$ through radius $\ell$, followed by shell filtering, returns exactly the fiber $\Lcal_s$.

Polynomial-time list decoding up to the Johnson radius originates with Guruswami and Sudan \cite{GuruswamiSudan1999}.  For a uniform reversible implementation with complexity polynomial in $\log p$, we use the following deterministic theorem.

\begin{theorem}[Chatterjee--Harsha--Kumar \cite{ChatterjeeEtAl2026}]
\label{thm:CHK}
For every finite field $\F$, block length $M$, and dimension $K<M$, there is a deterministic algorithm running in $\poly(M,\log\abs{\F})$ bit operations that outputs every Reed--Solomon codeword agreeing with a received word in more than $\sqrt{(K-1)M}$ coordinates.
\end{theorem}

\begin{corollary}[Complete decoding of the dual OPI code]
\label{cor:dual-list-decoding}
There is a deterministic algorithm that, on input $s\in\F_p^n$, outputs exactly $\Lcal_s$ in time $\poly(m,\log p)$ whenever
\begin{equation}
   m-\ell>\sqrt{(m-n-1)m}.
   \label{eq:finite-Johnson}
\end{equation}
The output list has polynomial size.
\end{corollary}

\begin{proof}
By \cref{lem:dual-GRS}, divide coordinate $i$ of $r_s$ by the nonzero multiplier $\nu_i$.  This maps $\C^\perp$ to the ordinary Reed--Solomon code of dimension $m-n$ and preserves Hamming distance.  Under \cref{eq:finite-Johnson}, every codeword at distance at most $\ell$ has agreement at least $m-\ell>\sqrt{(m-n-1)m}$, so \cref{thm:CHK} outputs all of them.  Rescale, form $y=r_s-c$, retain exactly those with $\wt(y)\in\Kcal$, remove duplicates, and sort lexicographically.  Every step is deterministic and polynomial time.  A polynomial-time machine can write only polynomially many output entries, which gives the list-size bound.
\end{proof}

\begin{lemma}[Explicit uniform list bound]
\label{lem:explicit-list-bound}
Let $\C$ be an $[m,k]$ generalized Reed--Solomon code, put $a=m-\ell$, and suppose $a^2>m(k-1)$.  Every Hamming ball of radius $\ell$ contains at most
\begin{equation}
   J(m,k,a):=
   \frac{a-(k-1)}{a^2/m-(k-1)}\le m^2
   \label{eq:explicit-list-bound}
\end{equation}
codewords.  If $k/m$ and $a/m$ converge while $a^2/m^2-k/m$ stays positive by a fixed constant, then $J(m,k,a)=O(1)$.
\end{lemma}

\begin{proof}
Since $a\le m$ and $a^2>m(k-1)$, we have $a>k-1$; in particular, the numerator and denominator in \cref{eq:explicit-list-bound} are positive.
Suppose that $N$ codewords lie in the ball.  For each codeword choose an $a$-element subset $S_u$ of coordinates on which it agrees with the received word, and put $t_i=\abs{\{u:i\in S_u\}}$.  Distinct generalized Reed--Solomon codewords agree in at most $k-1$ coordinates, so
\[
   \sum_i\binom{t_i}{2}
   =\sum_{u<v}\abs{S_u\cap S_v}
   \le\binom N2(k-1).
\]
Since $\sum_it_i=Na$, Cauchy--Schwarz gives
\[
   \sum_i\binom{t_i}{2}
   \ge\frac12\left(\frac{N^2a^2}{m}-Na\right).
\]
Rearranging yields
$N(a^2/m-(k-1))\le a-(k-1)$.  The displayed bound also holds when $N=1$ because
$J-1=a(1-a/m)/(a^2/m-(k-1))\ge0$.  Finally, $a^2-m(k-1)$ is a positive integer, so the denominator in \cref{eq:explicit-list-bound} is at least $1/m$, while its numerator is at most $m$.  Under a fixed normalized margin, the numerator and denominator are both linear in $m$, proving the last assertion.
\end{proof}

\section{Coherent Fiber Summation}

\subsection{A Canonical Fixed-Length List Representation}

Fix one deterministic implementation of \cref{cor:dual-list-decoding}, including lexicographic sorting and duplicate removal.  Apply \cref{lem:explicit-list-bound} with $k=m-n$ and $a=m-\ell$, and set
\begin{equation}
   B_\ell:=\left\lceil J(m,m-n,m-\ell)\right\rceil\le m^2,
   \qquad
   L:=2^{\lceil\log_2B_\ell\rceil}<2m^2,
   \qquad q:=\log_2L.
   \label{eq:global-L}
\end{equation}
Represent the sorted fiber as
\[
   \Lcal_s=(y_{s,0},\ldots,y_{s,N_s-1})
\]
and pad slots $N_s,\ldots,L-1$ with a validity bit equal to zero.  The same $L$ is used for every syndrome.  At a fixed strict asymptotic Johnson margin, both $B_\ell$ and $L$ are constant in $m$.

\begin{lemma}[Canonical reversible indexing]
\label{lem:index-unitary}
There is a uniform quantum circuit whose size is polynomial in $m$ and $\log p$.  Its action on every supported basis state is
\begin{equation}
   \ket{y_{s,j}}\ket{s}\ket{0^q}_J\ket0_Q\ket{0^{\mathrm{anc}}}
   \longmapsto
   \ket{\mathbf 0}_{\F_p^m}\ket{s}\ket{j}_J\ket1_Q\ket{0^{\mathrm{anc}}},
   \qquad 0\le j<N_s.
   \label{eq:index-map}
\end{equation}
Here $Q$ is a match flag.  The circuit is a unitary on the full computational basis, including inputs that are not valid fiber elements or valid field encodings.
\end{lemma}

\begin{proof}
Pad the deterministic decoder and its sorting routine to a fixed polynomial number of steps.  Bennett's compute--copy--uncompute construction turns this fixed-time computation into a reversible circuit with polynomial overhead \cite{Bennett1973}.  On input $s$, compute the padded canonical list and its validity bits into workspace $W$.  Reversibly compare $y$ with every valid slot, compute the first matching index into $J$, and set $Q=1$ exactly when a match exists.  When no slot matches, use the harmless convention $J=0,Q=0$.  Uncompute all comparison scratch space before changing the error register.  Thus a genuine match in slot zero is distinguished from a no-match input.

Controlled on $Q=1$, use $(W,J)$ to compute the selected list entry into a temporary register, subtract it from the error register, and reverse the selection computation to clear the temporary register.  Finally reverse the decoder computation to clear $W$; its input $s$ has not changed.  If $y=y_{s,j}$ is valid, the error register is zero and the retained pair is $(J,Q)=(j,1)$.  On a no-match or invalid-encoding input the controlled subtraction is skipped.  Every step is a reversible permutation of basis states with specified behavior on all invalid inputs, so the circuit is a unitary on the entire Hilbert space and \cref{eq:index-map} holds on the support of \cref{eq:labeled-state}.
\end{proof}

\subsection{Uniform List-Index Projection}

Apply \cref{lem:index-unitary} to \cref{eq:labeled-state}.  The result is
\begin{equation}
   \frac1{\sqrt{Z_0}}
   \sum_s\sum_{j<N_s}a_{y_{s,j}}\ket{\mathbf 0}_{\F_p^m}\ket{s}\ket{j}_J\ket1_Q.
   \label{eq:indexed-state}
\end{equation}
Apply $\Had^{\otimes q}$ to $J$ and measure $(J,Q)$ in the computational basis.

\begin{theorem}[Exact coherent fiber sum]
\label{thm:exact-fiber-sum}
If $Z_I>0$, then conditioned on the joint outcome $J=0^q,Q=1$, the syndrome register in \cref{eq:indexed-state} is exactly $\ket{\widehat I}$ from \cref{eq:ideal-syndrome}.  The success probability is exactly
\begin{equation}
   p_{\mathrm{fib}}=\frac{Z_I}{LZ_0}.
   \label{eq:success-exact}
\end{equation}
After the inverse QFT over $\F_p^n$, measurement samples exactly from $P_u$.
\end{theorem}

\begin{proof}
For every $j\in\{0,\ldots,L-1\}$,
\[
   \langle0^q\vert \Had^{\otimes q}\vert j\rangle=L^{-1/2}.
\]
Hence the unnormalized successful branch is
\begin{align*}
   \frac1{\sqrt{LZ_0}}\sum_s\sum_{j<N_s}a_{y_{s,j}}\ket{s}
   &=\frac1{\sqrt{LZ_0}}\sum_sA_s\ket{s}.
\end{align*}
Its squared norm is $Z_I/(LZ_0)$, and normalization gives \cref{eq:ideal-syndrome}.  This is a valid probability: by Cauchy--Schwarz in each fiber,
\[
   Z_I=\sum_s\abs{\sum_{y\in\Lcal_s}a_y}^2
      \le\sum_sN_s\sum_{y\in\Lcal_s}\abs{a_y}^2
      \le LZ_0.
\]
The final claim follows from \cref{lem:fourier-expansion}.
\end{proof}

\begin{remark}[Uniform padding across lists]
A syndrome-dependent projection length $N_s$ would produce the state
$\sum_s A_sN_s^{-1/2}\ket{s}$.  Padding every list to one global length $L$ replaces these varying factors by the common factor $L^{-1/2}$ and yields the Sun--Wootters state.
\end{remark}

\subsection{Algorithm Statement}

One sampling attempt consists of the following operations.
\begin{enumerate}[label=\textbf{\arabic*.},leftmargin=2.2em]
  \item Prepare the labeled state $\ket{\Psi_0}$ in \cref{eq:labeled-state}.
  \item Apply the reversible canonical indexing circuit of \cref{lem:index-unitary}.
  \item Apply $\Had^{\otimes q}$ to the list-index register and measure $(J,Q)$.  Restart unless $(J,Q)=(0^q,1)$.
  \item Apply $(\QFT_p^{-1})^{\otimes n}$ to the syndrome register and measure $x\in\F_p^n$.
\end{enumerate}

The next section proves that all four steps have polynomial complexity in a standard random-access model and that the restart probability is polynomially bounded under the Sun--Wootters condition.

\section{State Preparation, Complexity, and Robustness}

\subsection{Preparing the Labeled Error State}

We first isolate the same local data-access primitive used by DQI.

\begin{definition}[Fourier-state access]
For each $i\in[m]$, let $U_i$ be a clean unitary satisfying
\begin{equation}
   U_i\ket{0}=\sum_{z\in\F_p}\widehat g_i(z)\ket{z}.
   \label{eq:local-fourier-state}
\end{equation}
Controlled applications of $U_i$ and $U_i^\dagger$ are available at cost $T_g$.
\end{definition}

\begin{proposition}[Labeled error-state preparation]
\label{prop:shell-prep}
Given Fourier-state access, the labeled state $\ket{\Psi_0}$ can be prepared with
$\poly(m,\log p)+O(mT_g)$ elementary operations and clean ancillas.
\end{proposition}

\begin{proof}
Prepare a uniform superposition over $k\in\Kcal$ using a standard reversible range-state construction, with all rotations synthesized to the precision required below. Conditioned on $k$, prepare the Dicke state
\[
   \binom mk^{-1/2}\sum_{\substack{t\in\{0,1\}^m\\\wt(t)=k}}\ket{t}.
\]
A uniform Dicke state can be prepared sequentially: if $N$ positions and $R$ ones remain, rotate the next bit with amplitudes $\sqrt{(N-R)/N}$ and $\sqrt{R/N}$ and update the counters.  Induction on $N$ shows that every weight-$k$ string receives amplitude $\binom mk^{-1/2}$; reversible fixed-point arithmetic gives polynomial gate complexity.

For each coordinate $i$, controlled on $t_i=1$, apply $U_i$ to a field register initially equal to zero.  Because $\widehat g_i(0)=0$, the resulting field value is nonzero exactly on the branch $t_i=1$.  Compute the predicate $\one\{y_i\ne0\}$ and xor it into $t_i$, clearing the whole support register.  Compute $\wt(y)$ reversibly, subtract it from the $k$ register, and then reverse the weight computation; this clears both the $k$ register and its counter workspace.  The amplitude of a resulting vector $y$ of weight $k$ is
\[
   \frac1{\sqrt{\abs{\Kcal}}}\binom mk^{-1/2}
   \prod_{i:y_i\ne0}\widehat g_i(y_i)
   =\frac{a_y}{\sqrt{Z_0}},
\]
where \cref{lem:Z0} was used.  Finally compute $Hy$ by reversible field arithmetic.  For finite precision, expand the range-state, Dicke-state, controlled local-state, and prime-QFT routines into a total of $G_{\mathrm{app}}=\poly(m,\log p,\log(1/\varepsilon))$ approximate primitive unitaries.  Synthesize each primitive to operator error at most $\varepsilon/G_{\mathrm{app}}$; the hybrid argument then bounds the total prepared-state error by $\varepsilon$.  Exact reversible arithmetic does not consume this error budget.
\end{proof}

\subsection{From Membership Access to Fourier-State Access}

Suppose the input is supplied through the standard random-access membership oracle
\begin{equation}
   O_S:\ket{i}\ket{a}\ket{b}
   \longmapsto
   \ket{i}\ket{a}\ket{b\oplus\one\{a\in S_i\}}.
   \label{eq:membership-oracle}
\end{equation}

\begin{proposition}[Balanced local-state preparation]
\label{prop:membership-to-fourier}
Assume $\rho\in[\rho_*,1-\rho_*]$ for a constant $\rho_*>0$ and that $\rho$ is known.  Given $O_S$ and $O_S^\dagger$, a controlled version of \cref{eq:local-fourier-state} can be implemented to operator-norm error at most $\varepsilon$ using $O_{\rho_*}(1)$ oracle calls and $\poly(\log p,\log(1/\varepsilon))$ additional gates.
\end{proposition}

\begin{proof}
Because every coordinate has the same known density $\rho$, the phase choices below are independent of $i$; the entire construction can therefore be controlled coherently by an index register.  Exact amplitude amplification with a known marked fraction prepares the uniform states
\[
   \ket{S_i}=\frac1{\sqrt{\rho p}}\sum_{a\in S_i}\ket{a},
   \qquad
   \ket{\overline S_i}=\frac1{\sqrt{(1-\rho)p}}\sum_{a\notin S_i}\ket{a}
\]
using $O(1/\sqrt\rho)$ and $O(1/\sqrt{1-\rho})$ membership queries, respectively \cite{BrassardEtAl2002}.  Prepare a branch qubit as
$\sqrt{1-\rho}\ket0-\sqrt\rho\ket1$ and, conditioned on the branch, prepare $\ket{S_i}$ or $\ket{\overline S_i}$.  The joint state is
\[
   \sum_{a\in S_i}\frac{\sqrt{1-\rho}}{\sqrt{\rho p}}\ket0\ket a
   -\sum_{a\notin S_i}\frac{\sqrt\rho}{\sqrt{(1-\rho)p}}\ket1\ket a.
\]
The branch is $0$ precisely on $S_i$ and $1$ precisely on its complement.  Using the branch as the oracle target maps both cases to $\ket1$; one Pauli $X$ then clears it to $\ket0$.  The remaining field register is
\[
   \frac1{\sqrt p}\sum_a g_i(a)\ket a.
\]
Applying the negative-phase QFT over $\mathbb Z_p$ gives \cref{eq:local-fourier-state} by \cref{eq:fourier-convention}.  Exact amplitude amplification uses efficiently specifiable phase rotations; approximating those rotations and the arbitrary-modulus QFT to total operator error $\varepsilon$ costs $\poly(\log p,\log(1/\varepsilon))$ gates \cite{HalesHallgren2000}.  The common density $\rho$ and the chosen phase schedule determine the branch phases.  A controlled one-qubit phase correction preserves the relative minus sign before the branches are combined.  The query count is constant because $\rho$ is bounded away from zero and one.
\end{proof}

\subsection{Polynomial Success Probability}

\begin{lemma}[Postselection probability]
\label{lem:norm-transfer}
Under the assumptions of \cref{thm:SW-import},
\begin{equation}
   \frac{Z_I}{Z_0}=1+O\!\left(\frac{m^Ce^{-cm}}{\sigma+1}\right)=1+o(1).
   \label{eq:norm-ratio}
\end{equation}
Consequently, for all sufficiently large $m$,
\begin{equation}
   p_{\mathrm{fib}}\ge\frac1{2L}.
   \label{eq:success-lower-bound}
\end{equation}
\end{lemma}

\begin{proof}
By \cref{lem:fourier-expansion}, $Z_I=\E_x\abs{F(x)}^2$.  Apply \cref{eq:SW-denominator} and \cref{lem:Z0}, then substitute into \cref{eq:success-exact}.
\end{proof}

Thus the required numbers of attempts in the ideal and finite-precision circuits are, respectively,
\[
   N_{\mathrm{ideal}}=\left\lceil2L\log(1/\beta)\right\rceil,
   \qquad
   N_{\mathrm{app}}=\left\lceil4L\log(1/\beta)\right\rceil.
\]
They produce a successful list-index projection with probability at least $1-\beta$, giving a standard bounded-error algorithm with restart.

\subsection{Finite-Precision Robustness}

\begin{lemma}[Conditioning stability]
\label{lem:conditioning-stability}
Let $\ket\psi$ and $\ket\phi$ be unit vectors with $\norm{\ket\psi-\ket\phi}\le\varepsilon$, and let $\Pi$ be a projector.  Write $p=\norm{\Pi\ket\psi}^2$ and $q=\norm{\Pi\ket\phi}^2$.  If $\varepsilon<\sqrt p$, then
\begin{align}
   \abs{p-q}&\le2\varepsilon,
   \label{eq:prob-stability}\\
   \norm{\frac{\Pi\ket\psi}{\sqrt p}-
          \frac{\Pi\ket\phi}{\sqrt q}}
      &\le\frac{2\varepsilon}{\sqrt p}.
   \label{eq:conditional-stability}
\end{align}
\end{lemma}

\begin{proof}
The reverse triangle inequality gives
$\abs{\sqrt p-\sqrt q}\le\norm{\Pi(\ket\psi-\ket\phi)}\le\varepsilon$.
Since $\sqrt p,\sqrt q\le1$, this implies \cref{eq:prob-stability}.  For the normalized vectors,
\begin{align*}
  \norm{\frac{\Pi\ket\psi}{\sqrt p}-
         \frac{\Pi\ket\phi}{\sqrt q}}
  &\le
    \frac{\norm{\Pi(\ket\psi-\ket\phi)}}{\sqrt p}
    +\sqrt q\abs{\frac1{\sqrt p}-\frac1{\sqrt q}}\\
  &\le\frac{\varepsilon}{\sqrt p}
       +\frac{\abs{\sqrt q-\sqrt p}}{\sqrt p}
   \le\frac{2\varepsilon}{\sqrt p}.
\end{align*}
\end{proof}

\begin{corollary}[Finite-precision implementation]
\label{cor:finite-precision}
Under \cref{eq:success-lower-bound}, an implementation whose premeasurement state is within
\begin{equation}
   \varepsilon\le
   \min\left\{\frac1{8L},\frac{\eta}{4\sqrt{2L}}\right\}
   \label{eq:precision-choice}
\end{equation}
of the ideal state immediately before the fiber-projection measurement, and whose final inverse QFT has operator error at most $\eta/2$, has success probability at least $1/(4L)$ and, conditioned on success, produces an output distribution at total variation distance at most $\eta$ from $P_u$.
\end{corollary}

\begin{proof}
By \cref{lem:conditioning-stability} and $p\ge1/(2L)$, the success probability decreases by at most $2\varepsilon\le1/(4L)$, while the Euclidean distance between successful pure states before the final QFT is at most $2\varepsilon\sqrt{2L}\le\eta/2$.  The synthesized inverse QFT adds at most $\eta/2$ in Euclidean norm.  The triangle inequality, followed by contraction of trace distance under measurement, gives output total variation distance at most $\eta$.
\end{proof}

All required ideal rotations, local state preparations, and QFTs can be synthesized to inverse-polynomial precision, so \cref{eq:precision-choice} incurs only polynomial overhead.  The list decoder, canonical sorting, comparisons, and finite-field arithmetic are discrete reversible computations and introduce no approximation error.

\section{Main Algorithmic Theorem}

\begin{theorem}[Efficient Sun--Wootters sampler below the Johnson radius]
\label{thm:main-sampler}
Consider a balanced prime-field OPI family with $n/m\to r=2\mu\in(0,1)$ and $\rho_m=1/2+o(1)$.  Let $\delta\in(0,1/2-\mu]$ and $\lambda$ satisfy the Sun--Wootters condition \cref{eq:SW-condition}, and fix accuracy parameters $\eta,\beta\in(0,1/2)$.  Suppose also that
\begin{equation}
   \mu+\delta<1-\sqrt{1-r}.
   \label{eq:asymptotic-Johnson}
\end{equation}
Assume that $\delta$ is a fixed rational number.  For each family member put
\begin{equation}
   \mu_m:=\frac{n}{2m},
   \qquad
   \ell_m:=\min\left\{
      \left\lfloor\frac n2+\delta m\right\rfloor,
      \left\lfloor\frac m2\right\rfloor
   \right\},
   \label{eq:computable-shell-family}
\end{equation}
and define $P_{u,m}$ by \cref{eq:shell,eq:weights,eq:Pu} with largest active weight $\ell_m$.  Then, with membership-oracle access \cref{eq:membership-oracle}, there is a uniform bounded-error quantum algorithm running in
\begin{equation}
   \poly(m,\log p,\log(1/\eta),\log(1/\beta))
\end{equation}
time and outputting either a failure symbol $\bot$ or a vector $x\in\F_p^n$ such that
\begin{equation}
   \Prb[x=\bot]\le\beta,
   \qquad
   \TV\!\left(\mathcal L(x\mid x\ne\bot),P_{u,m}\right)\le\eta.
   \label{eq:sampler-failure-semantics}
\end{equation}
If $\bot$ is replaced by any fixed default vector, the unconditional output law is within total variation distance $\eta+\beta$ of $P_{u,m}$, and hence
\begin{equation}
   \E[s(x)]\ge\SCL_{1/2}(\mu+\delta)-o(1)-\eta-\beta.
   \label{eq:algorithm-score}
\end{equation}
\end{theorem}

\begin{proof}
We have $\mu_m\to\mu$ and $\ell_m/m\to\mu+\delta$: the cap is eventually inactive when $\mu+\delta<1/2$, and both entries in the minimum converge to $1/2$ at the endpoint.  In particular, $\ell_m\le m/2$ at every finite length, as required by the balanced Sun--Wootters estimate.  The strict exponent margin survives by \cref{lem:SW-stability}.  Condition \eqref{eq:asymptotic-Johnson} has constant slack, so the finite inequality \cref{eq:finite-Johnson} holds for all sufficiently large $m$.  By \cref{cor:dual-list-decoding,lem:explicit-list-bound}, the complete canonical fiber list is computable in polynomial time, $L<2m^2$, and in fact $L=O_{r,\delta}(1)$.  By \cref{prop:membership-to-fourier,prop:shell-prep}, the labeled state is preparable in polynomial time to any inverse-polynomial error.  By \cref{thm:exact-fiber-sum}, a successful joint projection $(J,Q)=(0^q,1)$ produces the exact ideal state in the ideal circuit.  By \cref{lem:norm-transfer}, the ideal attempt succeeds with probability at least $1/(2L)$; after approximation, \cref{cor:finite-precision} gives success at least $1/(4L)$.  Run
$N=\lceil4L\log(1/\beta)\rceil$ independent attempts, output the first successful sample, and output $\bot$ only if all $N$ attempts fail.  Since $(1-1/(4L))^N\le e^{-N/(4L)}\le\beta$, the failure bound in \cref{eq:sampler-failure-semantics} follows.  Conditioned on at least one success, the first successful attempt has the same successful-branch law as a single attempt, and \cref{cor:finite-precision} bounds that law by $\eta$ in total variation distance from $P_{u,m}$.

Replacing $\bot$ by a fixed vector changes the conditional successful law by mixing in mass at most $\beta$, so the resulting unconditional law is within $\eta+\beta$ of $P_{u,m}$.  Equation \eqref{eq:algorithm-score} follows from \cref{eq:SW-score,lem:SW-stability}, because a function taking values in $[0,1]$ changes its expectation by at most total variation distance.
\end{proof}

\begin{corollary}[High-probability optimization output]
\label{cor:high-probability-score}
Under \cref{thm:main-sampler}, for every fixed $\varepsilon>0$ and $\beta\in(0,1)$, polynomially many independent samples followed by classical score evaluation output, with probability at least $1-\beta$, an $x$ satisfying
\begin{equation}
   s(x)\ge\SCL_{1/2}(\mu+\delta)-\varepsilon-o(1).
\end{equation}
\end{corollary}

\begin{proof}
Let $t=\SCL_{1/2}(\mu+\delta)$.  Choose $\eta$ and the per-call failure probability so that, after replacing a possible failure by a default vector, one invocation has expectation at least $t-\varepsilon/2-o(1)$.  If $P$ is the probability that $s(x)\ge t-\varepsilon$, then
\[
   \E[s(x)]\le P+(1-P)(t-\varepsilon),
\]
so for large $m$,
\[
   P\ge\frac{\varepsilon/3}{1-t+\varepsilon},
\]
a positive constant depending only on the fixed parameters.  Repetition and selection of the best observed score amplify this to $1-\beta$.  Each score is evaluated with $m$ membership-oracle calls and polynomial-time finite-field arithmetic.
\end{proof}

\section{From Existence to a Strict Algorithmic Separation}

The main theorem requires the target radius to be below the Johnson radius.  We now show that this restriction never prevents a \emph{strict} improvement when the Sun--Wootters exponent condition permits one.

\begin{lemma}[Strict gap between unique and Johnson radii]
\label{lem:johnson-gap}
For $r\in(0,1)$ and $\mu=r/2$,
\begin{equation}
   1-\sqrt{1-r}-\mu
   =\frac{(1-\sqrt{1-r})^2}{2}>0.
   \label{eq:johnson-gap}
\end{equation}
\end{lemma}

\begin{proof}
Put $t=\sqrt{1-r}$.  Then $r=1-t^2$ and
\[
   1-t-\frac{1-t^2}{2}=\frac{(1-t)^2}{2}.
\]
\end{proof}

\begin{lemma}[Monotonicity of the feasibility exponent]
\label{lem:E-monotone}
For fixed $\mu\in(0,1/2)$, the function $\delta\mapsto E(\mu,\delta)$ is strictly increasing on $0<\delta<1/2-\mu$.
\end{lemma}

\begin{proof}
Differentiating \cref{eq:E-exponent} and using $h'(x)=\log(1/x-1)$ gives
\begin{equation}
   \frac{\partial E}{\partial\delta}
   =\log\!\left(
      \frac{(\mu+\delta)(1-2\mu-\delta)}
           {\delta(1-\mu-\delta)}
      \right).
   \label{eq:E-derivative}
\end{equation}
The numerator inside the logarithm exceeds the denominator because their difference is
\[
   (\mu+\delta)(1-2\mu-\delta)
   -\delta(1-\mu-\delta)
   =\mu(1-2\mu-2\delta)>0.
\]
\end{proof}

\begin{lemma}[Endpoint bounds]
\label{lem:endpoint-bounds}
The two parameter choices used below satisfy
\begin{align}
 E\!\left(\frac{249}{800},10^{-6}\right)
 +G\!\left(\frac{249}{800},\frac{3429}{20000}\right)
 &< -\frac1{10000},
 \label{eq:base-rate-bound}\\
 E\!\left(\frac38,\frac18\right)
 +G\!\left(\frac38,\frac{99}{250}\right)
 &< -\frac1{2000}.
 \label{eq:three-quarter-bound}
\end{align}
\end{lemma}

\begin{proof}
We bound every logarithm and $\pi$ by rational intervals.  For a rational
$y\in[1,2]$, set $z=(y-1)/(y+1)$.  The positive series for
$2\operatorname{arctanh}z$ gives, for every $N\ge1$,
\begin{equation}
  2\sum_{j=0}^{N-1}\frac{z^{2j+1}}{2j+1}
  \le \log y
  \le
  2\sum_{j=0}^{N-1}\frac{z^{2j+1}}{2j+1}
  +\frac{2z^{2N+1}}{(2N+1)(1-z^2)}.
  \label{eq:log-rational-bound}
\end{equation}
Writing any positive rational $x$ as $x=2^k y$ with $y\in[1,2)$ and applying
\cref{eq:log-rational-bound} to $y$ and to $2$ gives rational bounds for
$\log x$.  All logarithm arguments in the two evaluations are rational except
for $2/\pi$.  We bound this remaining argument by first enclosing $\pi$ through
Machin's identity
\[
   \frac\pi4=4\arctan\frac15-\arctan\frac1{239},
\]
using the alternating Taylor series with its signed next-term remainder.
Taking $N=100$ in \cref{eq:log-rational-bound}, $40$ terms for
$\arctan(1/5)$, and $16$ terms for $\arctan(1/239)$, and propagating the
resulting intervals with the appropriate signs through the definitions of
$E$ and $G$, gives the rational upper bounds
\[
 -1.41885473\times10^{-4}
 \quad\text{and}\quad
 -8.18554771\times10^{-4},
\]
respectively.  They are smaller than $-1/10000$ and $-1/2000$, proving
\cref{eq:base-rate-bound,eq:three-quarter-bound}.
\end{proof}

\begin{lemma}[Feasibility throughout the rate interval]
\label{lem:global-feasibility}
For every
\[
   \frac{249}{800}\le \mu<\frac12,
\]
there are parameters $\delta_0\in(0,1/2-\mu]$ and $\lambda$ in the entropy domain of $G$ such that
\[
   E(\mu,\delta_0)+G(\mu,\lambda)<0.
\]
\end{lemma}

\begin{proof}
We prove the assertion analytically from the first endpoint bound in
\cref{lem:endpoint-bounds}.  Put
\[
  \Phi(\mu):=E(\mu,0)+\min_{\lambda\in I_\mu}G(\mu,\lambda),
  \qquad
  I_\mu:=\bigl(\max\{0,6\mu-2\},\,4\mu-1\bigr).
\]
For $\mu\in[249/800,1/2)$, the interval $I_\mu$ is precisely the interior domain in which both entropy arguments depending on $\lambda$ lie in $(0,1)$.  Write
\[
   a=4\mu-1,
   \quad c=2-4\mu,
   \quad d=2\mu-\lambda,
   \quad e=2-6\mu+\lambda,
   \quad \kappa=\frac2\pi.
\]
A direct differentiation gives
\[
   \frac{\partial^2G}{\partial\lambda^2}
    =\frac1\lambda+\frac1{a-\lambda}
      +\frac1d+\frac1e>0.
\]
Moreover, $\partial G/\partial\lambda$ tends to $-\infty$ at the left endpoint of $I_\mu$ and to $+\infty$ at the right endpoint.  Thus there is a unique interior minimizer $\lambda_*(\mu)$, depending smoothly on $\mu$, and its stationarity equation is
\begin{equation}
   \kappa\lambda_* e=(a-\lambda_*)d.
   \label{eq:lambda-stationarity}
\end{equation}
The envelope theorem, or direct differentiation followed by
\cref{eq:lambda-stationarity}, yields
\begin{equation}
  \Phi'(\mu)=\log R,
  \qquad
  R=\frac{\mu}{1-\mu}
       \left(\frac{1-2\mu}{2\mu}\right)^2
       (1-x)^4\frac{z^2}{(1-z)^6},
  \label{eq:Phi-derivative-ratio}
\end{equation}
where
\[
   x:=\frac{\lambda_*}{a},
   \qquad
   z:=\frac{d}{c}.
\]
We now bound the right side uniformly by a rational number below one.

Set $t=1-2\mu$, so $0<t\le151/400$, and let
\[
   u=1-x,
   \qquad v=1-z,
   \qquad D=\kappa+(1-\kappa)u,
   \qquad w=\frac ut.
\]
Equation~\eqref{eq:lambda-stationarity} gives $v=u/D$.  The identity
$ax+cz=2\mu=1-t$ then becomes
\begin{equation}
   L_{t,\kappa}(w):=
   w(1-2t)+\frac{2tw}{\kappa+(1-\kappa)tw}=1.
   \label{eq:Ltw}
\end{equation}
The left side is strictly increasing in $w$.  Since its denominator at $w=1$ is at most one,
$L_{t,\kappa}(1)\ge1$, so $w\le1$.  The classical bounds
$3<\pi<22/7$ imply $7/11<\kappa<2/3$.  The function
$L_{t,\kappa}(3/4)$ decreases with $\kappa$, and direct rational simplification gives
\[
  1-L_{t,7/11}(3/4)
   =\frac{18t^2-21t+7}{4(3t+7)}>0;
\]
the numerator is positive because its discriminant is $-63$.  Hence
\begin{equation}
   \frac34<w\le1,
   \qquad
   0<u=wt\le\frac{151}{400},
   \qquad
   D\le\frac23+\frac13\frac{151}{400}=\frac{317}{400}.
   \label{eq:wD-bounds}
\end{equation}
Using $v=u/D$, $z=\kappa x/D$, and $u=wt$, the ratio in
\cref{eq:Phi-derivative-ratio} simplifies exactly to
\[
   R=\frac{\kappa^2x^2D^4}{w^2(1-t^2)}.
\]
Therefore
\begin{equation}
   R\le
   \frac{(2/3)^2(317/400)^4}
        {(3/4)^2\bigl(1-(151/400)^2\bigr)}
   =\frac{10098039121}{27782797500}
   <1.
   \label{eq:global-R-bound}
\end{equation}
Thus $\Phi$ is strictly decreasing on the whole interval.

It remains to determine its sign at the left endpoint
$\mu_0=249/800$.  By \cref{eq:base-rate-bound},
\begin{equation}
 E\!\left(\frac{249}{800},10^{-6}\right)
 +G\!\left(\frac{249}{800},\frac{3429}{20000}\right)
 <-\frac1{10000}<0.
\end{equation}
By \cref{lem:E-monotone}, replacing $10^{-6}$ by $0$ only decreases the left side; taking the minimum over $\lambda$ decreases it again.  Hence
$\Phi(249/800)<0$, and monotonicity gives $\Phi(\mu)<0$ for every
$\mu\ge249/800$.  Finally, continuity in $\delta$ supplies a sufficiently small
$\delta_0>0$, also chosen below $1/2-\mu$, for which the strict inequality remains true.
\end{proof}

\begin{theorem}[Algorithmic transfer of strict Sun--Wootters improvements]
\label{thm:qualitative-transfer}
Fix $r=2\mu\in(0,1)$.  Suppose there are $\delta_0>0$ and $\lambda$ satisfying
\begin{equation}
   E(\mu,\delta_0)+G(\mu,\lambda)<0.
   \label{eq:feasible-delta0}
\end{equation}
Then there exists a fixed positive rational $\delta$ for which \cref{thm:main-sampler} applies.  Its expected score is strictly larger than the semicircle benchmark $\SCL_{1/2}(\mu)$ by a constant depending only on $r$ and the feasible parameters.
\end{theorem}

\begin{proof}
By density of the rationals and \cref{lem:johnson-gap}, choose a rational number satisfying
\begin{equation}
   0<\delta<\min\left\{
       \delta_0,
       1-\sqrt{1-r}-\mu
   \right\}.
   \label{eq:shrunk-delta}
\end{equation}
By \cref{lem:E-monotone},
$E(\mu,\delta)+G(\mu,\lambda)<E(\mu,\delta_0)+G(\mu,\lambda)<0$.
By \cref{lem:johnson-gap}, $\mu+\delta<1-\sqrt{1-r}$.  Therefore \cref{thm:main-sampler} applies.  Finally, \cref{eq:scl} is strictly increasing on $[0,1/2)$, so the constant
\[
   \Delta:=\SCL_{1/2}(\mu+\delta)-\SCL_{1/2}(\mu)
\]
is positive.  Invoke \cref{thm:main-sampler} with $\eta=\beta=\Delta/8$.  For all sufficiently large $m$, its $o(1)$ term is at most $\Delta/4$, and \cref{eq:algorithm-score} yields
\[
   \E[s(x)]
   \ge \SCL_{1/2}(\mu+\delta)-\frac{\Delta}{2}
   =\SCL_{1/2}(\mu)+\frac{\Delta}{2}.
\]
This is the claimed constant strict improvement.
\end{proof}

The left endpoint in \cref{lem:global-feasibility} is $2\mu=249/400=0.6225$.  Combining that lemma with \cref{thm:qualitative-transfer} proves the algorithmic separation for every fixed rate from $0.6225$ onward, matching the threshold reported by Sun and Wootters \cite{SunWootters2026}.

\begin{corollary}[Strict semicircle improvement from rate $0.6225$]
\label{cor:rate-06225}
Fix a rate
\begin{equation}
   0.6225\le r<1,
\end{equation}
and a sequence of prime-field Reed--Solomon OPI instances with $m\to\infty$, $n/m\to r$, $m\le p$ prime, distinct evaluation points, common density $\rho_m=1/2+o(1)$, and coherent membership-oracle access.  There is a constant $\gamma_r>0$ and a uniform bounded-error polynomial-time quantum algorithm whose expected satisfaction ratio is at least
\begin{equation}
   \SCL_{1/2}(r/2)+\gamma_r
\end{equation}
for all sufficiently large $m$.  Moreover, for every fixed $\varepsilon>0$ and $\beta\in(0,1)$, polynomially many samples output an $x$ satisfying
\begin{equation}
   s(x)\ge \SCL_{1/2}(r/2)+\gamma_r-\varepsilon
\end{equation}
with probability at least $1-\beta$, for all sufficiently large $m$.
\end{corollary}

\begin{proof}
Put $\mu=r/2$.  Then $\mu\ge249/800$, so \cref{lem:global-feasibility} supplies parameters satisfying \cref{eq:feasible-delta0}.  Let $\delta>0$ be the parameter selected in \cref{thm:qualitative-transfer}, and put
\[
   \Delta_r:=\SCL_{1/2}(\mu+\delta)-\SCL_{1/2}(\mu)>0,
   \qquad \gamma_r:=\frac{\Delta_r}{4}.
\]
The proof of \cref{thm:qualitative-transfer} gives expected score at least $\SCL_{1/2}(\mu)+\Delta_r/2$ for all sufficiently large $m$, which implies the first claim.  Apply \cref{cor:high-probability-score} to the same $\delta$, with its accuracy parameter chosen smaller than the stated $\varepsilon$.  Since its target is $\SCL_{1/2}(\mu)+\Delta_r$ and its $o(1)$ term vanishes, the resulting score is at least $\SCL_{1/2}(\mu)+\gamma_r-\varepsilon$ for all sufficiently large $m$; repetition gives failure probability at most $\beta$.
\end{proof}

\subsection{Asymptotically Perfect Solutions at Three-Quarter Rate}

At the limiting rate $3/4$, the asymptotic Johnson radius is exactly $1/2$.  To handle arbitrary floor effects, and even families whose finite rates approach $3/4$ from below, we choose the largest integer cutoff that is simultaneously at most $m/2$ and strictly inside the finite decoder radius.

\begin{lemma}[Finite Johnson endpoint cutoff]
\label{lem:three-quarter-slack}
Let $1\le n'<m$ and define
\begin{equation}
   \ell_J(m,n'):=
   \min\left\{
      \left\lfloor\frac m2\right\rfloor,
      \left\lceil m-\sqrt{(m-n'-1)m}\right\rceil-1
   \right\}.
   \label{eq:endpoint-cutoff}
\end{equation}
Then
\begin{equation}
   m-\ell_J(m,n')>\sqrt{(m-n'-1)m}.
   \label{eq:endpoint-finite-Johnson}
\end{equation}
If $n'/m\to3/4$, then $\ell_J(m,n')/m\to1/2$.
\end{lemma}

\begin{proof}
For every real $a$, $\lceil a\rceil-1<a$.  Thus the second entry in the minimum in \cref{eq:endpoint-cutoff} is strictly smaller than
$m-\sqrt{(m-n'-1)m}$, proving \cref{eq:endpoint-finite-Johnson}.  If $n'/m\to3/4$, then
\[
   \frac1m\left(m-\sqrt{(m-n'-1)m}\right)
   =1-\sqrt{1-\frac{n'}m-\frac1m}\longrightarrow\frac12.
\]
Both entries in the minimum therefore have ratio tending to $1/2$.
\end{proof}

The second endpoint bound in \cref{lem:endpoint-bounds} applies at
\begin{equation}
   \mu=\frac38,
   \qquad \delta=\frac18,
   \qquad \lambda=\frac{99}{250},
   \label{eq:three-quarter-point}
\end{equation}
and shows that the Sun--Wootters exponent is strictly negative.  Hence the
strict condition persists in a neighborhood of this point.

\begin{corollary}[Algorithmic asymptotic perfection]
\label{cor:perfect-three-quarter}
Consider any sequence of prime-field Reed--Solomon OPI instances with $m\to\infty$, $m\le p$ prime, distinct evaluation points, common density $\rho_m=1/2+o(1)$, and coherent membership-oracle access.  If $n/m\to r\ge3/4$, there is a uniform polynomial-time quantum algorithm that outputs $x$ with
\begin{equation}
   s(x)=1-o(1)
\end{equation}
with probability $1-o(1)$.
\end{corollary}

\begin{proof}
Set
\[
   n'_m:=\min\left\{n,\left\lfloor\frac{3m}{4}\right\rfloor\right\}.
\]
For all sufficiently large $m$, this integer satisfies $1\le n'_m<m$.  Since $n/m\to r\ge3/4$, we have $n'_m/m\to3/4$.  Restrict the coefficient space to the first $n'_m$ monomials; every polynomial produced for this subproblem is valid for the original degree bound because $n'_m\le n$.

Use the cutoff $\ell_m=\ell_J(m,n'_m)$ from \cref{lem:three-quarter-slack}, and set
\[
   \mu_m=\frac{n'_m}{2m},
   \qquad
   \delta_m=\frac{\ell_m}{m}-\mu_m.
\]
Then $(\mu_m,\delta_m)\to(3/8,1/8)$, while the finite list-decoding inequality holds for every $m$.  The strict bound in \cref{eq:three-quarter-bound}, continuity of $E+G$, and \cref{lem:SW-stability} imply
\[
   \E_{x\sim P_{u,m}}s(x)
      \ge \SCL_{1/2}\!\left(\frac{\ell_m}{m}\right)-o(1)
      =1-o(1).
\]
Here $P_{u,m}$ is the Sun--Wootters distribution defined by the cutoff $\ell_m$.  The coherent list-index construction and the deterministic list decoder therefore sample this distribution in polynomial time.  Choose circuit error and restart failure probability to be $o(1)$, for example $m^{-2}$.  If $\E[1-s(x)]\le\varepsilon_m=o(1)$, Markov's inequality gives
\[
   \Prb[s(x)<1-\sqrt{\varepsilon_m}]
      \le\sqrt{\varepsilon_m},
\]
which proves the claim.
\end{proof}

\begin{remark}[Sampling beyond the Johnson radius]
Horinaga and Yamakawa give a balanced existential saturation guarantee for $R>0.7158$ and an algorithm finding satisfaction $1$ at every fixed rate $R>3/4$ \cite{HorinagaYamakawa2026}.  For the Sun--Wootters distribution, the target radius $1/2$ in the range $0.7496\le r<0.75$ lies beyond the dual Reed--Solomon Johnson radius.  The sampler of \cref{thm:main-sampler} gives the strict semicircle improvement of \cref{cor:rate-06225} throughout this interval.  Exact sampling at radius $1/2$ would require coherent weighted summation beyond efficient complete list decoding.
\end{remark}

\section{Conclusion}

Coherent list-index projection converts a complete classical list decoder into an exact sum over errors sharing a syndrome.  Applied to the dual generalized Reed--Solomon code of OPI, it prepares the Sun--Wootters distribution for cutoffs satisfying the Johnson-radius bound.  The Sun--Wootters denominator estimate and the monotonicity argument yield a polynomial-time improvement from balanced rate $0.6225$ and asymptotically perfect satisfaction at rate $3/4$.

Concurrent work of Horinaga and Yamakawa gives a Regev-like algorithm over prime-power fields $q=p^e$ satisfying $p=\omega(1)$ and $\log q\le\poly(m)$, and extends to MDS MaxLINSAT instances whose dual code admits an efficient list decoder.  In balanced OPI, it finds solutions of satisfaction $1$ for every fixed rate above $3/4$.
\bibliographystyle{alpha}
\bibliography{references}

\end{document}